\documentclass{gAPA2e}

\begin{document}
\doi{10.1080/0003681YYxxxxxxxx}
 \issn{1563-504X}
\issnp{0003-6811}
\jvol{00} \jnum{00} \jyear{2009} \jmonth{July}

\markboth{Michal Fe\v ckan and Vassilis M. Rothos}{Traveling Waves of DNLS}

\def\C{\mathbb{C}}
\def\Q{\mathbb{Q}}
\def\H{\rm H}
\def\N{\mathbb{N}}
\def\R{\mathbb{R}}
\def\Z{\mathbb{Z}}
\def\I{\bm{I}}
\def\wt{\widetilde}
\def\th{\theta}
\def\la{\lambda}
\def\ds{\displaystyle}
\def\ep{\varepsilon}
\def\eu{ {\, \textrm{\rm e} }}
\def\om{\omega}
\def\RR{{\mathcal R}}
\def\KK{{\mathcal K}}
\def\FF{{\mathcal F}}
\def\LL{{\mathcal L}}
\def\ol{\overline}
\def\al{\alpha}
\def\be{\beta}

\title{\bf Traveling Waves of Discrete Nonlinear Schr\"odinger Equations with Nonlocal Interactions}
\author{Michal Fe\v ckan$^{\dagger}$$^{\ast}$\thanks{$^\ast$Corresponding author. Email: Michal.Feckan@fmph.uniba.sk
\vspace{6pt}} and Vassilis M. Rothos$^{\dagger \dagger}$\\\vspace{6pt}  $^{\dagger}${\em{Department of Mathematical Analysis and
Numerical Mathematics, Comenius University, Mlynsk\'a dolina, 842 48 Bratislava,
Slovakia, and Mathematical Institute of Slovak Academy of Sciences,
\v{S}tef\'anikova 49, 814 73 Bratislava, Slovakia}}; \\ $^{\dagger \dagger}${\em{School of Mathematics, Physics and Computational Sciences, Faculty of Engineering,
Aristotle University of Thessaloniki, Thessaloniki 54124, Greece}}\\\vspace{6pt}\received{v3.3 released July 2009} }

\maketitle

\begin{abstract}
Existence and bifurcation results are derived for quasi periodic traveling waves of discrete nonlinear Schr\"odinger
equations with nonlocal interactions and with polynomial type potentials. Variational tools are used. Several concrete nonlocal
interactions are studied as well.
\end{abstract}

\begin{classcode} 34K14, 37K60, 37L60\end{classcode}\bigskip

\section{Introduction}

One of the most exciting areas in applied mathematics is the study
of the dynamics associated with the propagation of information.
Coherent structures like solitons, kinks, vortices, etc., play a
central role, as carriers of energy, in many nonlinear physical
systems~\cite{book}. Solitons represent a rare example of a
(relatively) recently arisen mathematical object which has found
successful high-technology applications~\cite{Has}. The nature of
the system dictates that the relevant and important effects occur
along one axial direction. Interplay between
nonlinearity and periodicity is the focus of recent studies in
different branches of modern applied mathematics and nonlinear
physics. Applications range from nonlinear optics, in the dynamics
of guided waves in inhomogeneous optical structures and photonic
crystal lattices, to atomic physics, in the dynamics of
Bose-Einstein condensate (BEC) droplets in periodic potentials,
and from condensed matter, in Josephson-junction ladders, to
biophysics, in various models of the DNA double strand. Analysis
and modeling of these physical situations are based on nonlinear
evolution equations derived from underlying physics equations,
such as nonlinear Maxwell equations with periodic coefficients
~\cite{SulemC99}. In particular, the systems of 2nd-order NLS
equations, both continuous and discrete, were applied in nonlinear
physics to study a number of experimental and theoretical
problems. Spatial non-locality of the nonlinear response is also
naturally present in the description of BECs where it represents
the finite range of the bosonic interaction. Demands on the
mathematics for techniques to analyze these models may best be
served by developing methods tailored to determining the local
behavior of solutions near these structures. The discreteness of
space i.e., the existence of an underlying spatial lattice is
crucial to the structural stability of these spatially localized
nonlinear excitations.

During the early years, studies of intrinsic localized modes were
mostly of a mathematical nature, but the ideas of localized modes
soon spread to theoretical models of many different physical
systems, and the discrete breather concept has been recently
applied to experiments in several different physics
subdisciplines. Most nonlinear lattice systems are not integrable
even if the partial differential equation (PDE) model in the
continuum limit is. While for many years spatially continuous
nonlinear PDE's and their localized solutions have received a
great deal of attention, there has been increasing interest in
spatially discrete nonlinear systems. Namely, the dynamical
properties of nonlinear systems based on the interplay between
discreteness, nonlinearity and dispersion (or diffraction) can
find wide applications in various physical, biological and
technological problems. Examples are coupled optical fibres
(self-trapping of light)~\cite{Cris,Sergej,coupled,Lenz}, arrays
of coupled Josephson junctions~\cite{Joseph}, nonlinear charge and
excitation transport in biological macromolecules, charge
transport in organic semiconductors~\cite{semi}.

Prototype models for such nonlinear lattices take the form of
various nonlinear lattices \cite{Aubry}, a particularly important
class of solutions of which are so called discrete breathers which
are homoclinic in space and oscillatory in time. Other questions
involve the existence and propagation of topological defects or
kinks which mathematically are heteroclinic connections between a
ground and an excited steady state. Prototype models here are
discrete version of sine-Gordon equations, also known as
Frenkel-Kontorova (FK) models, e.g.\ \cite{ACR03}. There are many
outstanding issues for such systems relating to the global
existence and dynamics of localized modes for general
nonlinearities, away from either continuum or anti-continuum
limits.

In the main part of the previous studies of the discrete NLS
models the dispersive interaction was assumed to be short-ranged
and a nearest-neighbor approximation was used. However, there
exist physical situations that definitely can not be described in
the framework of this approximation. The DNA molecule contains
charged groups, with long-range Coulomb interaction $1/r$ between
them. The excitation transfer in molecular crystals \cite{Davydov}
and the vibron energy transport in biopolymers \cite{Scott} are
due to transition dipole-dipole interaction with $1/r^3$
dependence on the distance, $r$. The nonlocal (long-range)
dispersive interaction in these systems provides the existence of
additional length-scale: the radius of the dispersive interaction.
We will show that it leads to the bifurcating properties of the
system due to both the competition between nonlinearity and
dispersion, and the interplay of long-range interactions and
lattice discreteness.

In some approximation the equation of motion is the nonlocal
discrete NLS
\begin{equation}\label{e0a}
\imath \dot u_n=\sum\limits_{m\neq
n}J_{n-m}(u_{n}-u_{m})+|u_n|^2u_n,\quad n\in\Z\, ,
\end{equation}
where the long-range dispersive coupling is taken to be either
exponentially $J_{n}=J{\rm e}^{-\beta |n|}$ with $\beta>0$, or algebraically
$J_{n}=J|n|^{-s}$ with $s>0$, decreasing with the distance $n$ between
lattice sites. In both cases the constant $J$ is normalized such
that $\sum_{n=1}^{\infty}J_{n}=1$, for all $\beta$ or $s$. The
parameters $\beta$ and $s$ are introduced to cover different
physical situations from the nearest-neighbor approximation
$(\beta\to\infty, s\to\infty)$ to the quadrupole-quadrupole
$(s=5)$ and dipole-dipole $(s = 3)$ interactions. The Hamiltonian
$H$ and the number of excitations $N$
\begin{equation}\label{e0b}
H=\frac{1}{2}\sum_{n,
m\in\Z}J_{n-m}|u_{n}-u_{m}|^{2}-\frac{1}{2}\sum_{n\in\Z}|u_{n}|^{4},\quad
{\rm and}\quad N=\sum_{n\in\Z}|u_{n}|^{2}
\end{equation}
are conserved quantities corresponding to the set of (\ref{e0a}).

It should be also noted that the derivation of a discrete equation
from the Gross-Pitaevskii equation produces at the intermediate
step a fully nonlocal discrete NLS equation for the coefficients
of the wave function expansion over the complete set of the
Wannier functions. Further reduction to the case of the only band
with the strong localization of the Wannier functions (the
tight-binding approximation) leads to the standard local DNLS
equation. Recently Abdullaev et al.~\cite{Abd:08} extended this
approach to the case of periodic nonlinearities and derived a
number of nonintegrable lattices with different nearest-neighbor
nonlinearities.

In this paper, we study the discrete nonlinear
Schr\"odinger equations on the lattice $\Z$ (DNLS) with nonlocal
interactions of forms
\begin{equation}\label{e1}
\imath \dot u_n=\sum\limits_{j\in\N}a_j\Delta_ju_{n}+f(|u_n|^2)u_n,\quad n\in\Z
\end{equation}
where $u_n\in\C$, $\Delta_ju_{n}:=u_{n+j}+u_{n-j}-2u_{n}$ are $1$-dimensional discrete Laplacians and it holds

\begin{itemize}
\item[(H1)] $f\in C(\R_+,\R)$ for $\R_+:=[0,\infty)$, $f(0)=0$ and $a_j\in\R$ with $\sum\limits_{j\in\N}|a_j|<\infty$.
Moreover, there are constants $s>0$, $\mu>1$, $c_1>0$, $c_2>0$ and $\bar r>0$ such that
$$
\begin{gathered}
|f(w)|\le c_1(w^s+1),\quad c_2(w^{s+1}-1)\le F(w),\quad
\mu F(w)-\bar r<f(w)w\end{gathered}
$$
for any $w\ge 0$, where $F(w)=\int\limits_0^wf(z)dz$. Furthermore,
$\limsup_{w\to 0_+}f(w)/w^{\wt s}<\infty$ for a constant $\wt
s>0$.
\end{itemize}
Of course we suppose that not all $a_j$ are zero. Note any polynomial $f(w)=p_1w+\cdots +p_sw^s$, $s\in\N$ with $p_s>0$
satisfies (H1). Furthermore, \eqref{e1} can be rewritten into a standard form
\begin{equation}\label{e1a}
\imath \dot u_n=\sum\limits_{m\ne n}a_{|m-n|}\left(u_{m}-u_n\right)+f(|u_n|^2)u_n,\quad n\in\Z.
\end{equation}
It is well known that \eqref{e1a} conserves two dynamical invariants
$$
\begin{gathered}
\sum\limits_{n\in\Z}|u_n|^2\quad -\textrm{the norm},\\
\sum\limits_{n\in\Z}\left[-\frac{1}{2}\sum\limits_{m\ne n}a_{|m-n|}\left|u_{m}-u_n\right|^2+F(|u_n|^2)\right]\quad -\textrm{the energy}.
\end{gathered}
$$

Differential equations with nonlocal interactions on lattices have
been studied in \cite{ALK,ACG,BCh1,BChCh1,BZh1,CGGJMR,CCK,LCh,MG},
while DNLS (discrete nonlinear Schr\"odinger) in \cite{CKMF,CCMK,CGGJMR,GFB,LPK}. Nowadays it is
clear that a large number of important models of various fields of
physics are based on DNLS type equations with several forms of
polynomial nonlinearities starting with the simplest
self-focusing cubic (Kerr) nonlinearity, then following with the
cubic onsite nonlinearity relevant for Bose-Einstein condensates,
then with more general discrete cubic nonlinearity in Salerno
model up to cubic-quintic ones (see \cite{CCMK} for more
references).

We are interested in the existence of traveling wave solutions $u_n(t)=U(n-\nu t)$ of \eqref{e1} with a quasi
periodic function $U(z)$, $z=n-\nu t$ and some $\nu\ne0$.

First, we introduce a function
$$
\Phi(x):=\frac{4}{x}\sum\limits_{j\in\N}a_j\sin^2\left[\frac{x}{2}j\right]\, .
$$
\begin{remark}\label{inrem1}Clearly $\Phi\in C(\R\setminus\{0\},\R)$, $\Phi$ is odd, $\Phi(2\pi k)=0$ for
any $k\in \Z\setminus\{0\}$, and $\Phi(x)\to 0$ as $|x|\to \infty$.
If $\sum\limits_{j\in\N}j|a_j|<\infty $ then $\Phi\in C(\R,\R)$ and if $\sum\limits_{j\in\N}j^2|a_j|<\infty $ then $\Phi\in C^1(\R,\R)$.
Consequently the range $\RR \Phi:=\Phi(\R\setminus\{0\})$ is either an interval $[-\bar R,\bar R]$ or $(-\bar R,\bar R)$ here with
possibility $\bar R=\infty$ (see Section \ref{examples} for concrete examples).\end{remark}

Now we can state the following existence result.

\begin{theorem}\label{th1} Let (H1) hold and $T>0$. Then for almost each $\nu\in\R\setminus\{0\}$ and any rational $r\in \Q\cap (0,1)$,
there is a nonzero periodic traveling wave solution $u_n(t)=U(n-\nu t)$ of \eqref{e1} with $U\in C^1(\R,\C)$ and such that
\begin{equation}\label{prop}
U(z+T)=\eu^{2\pi r\imath}U(z),\, \forall z\in \R\, .
\end{equation}
Moreover, for any $\nu \in \R\setminus\{0\}$ there is at most a finite number of $\bar r_1,\bar r_2,\cdots,\bar r_m\in (0,1)$
such that equation
$$
-\nu=\Phi\left(\frac{2\pi}{T}(\bar r_j+k)\right)
$$
has a solution $k\in\Z$. Then for any $r\in (0,1)\setminus \{\bar r_1,\bar r_2,\cdots,\bar r_m\}$ there is a nonzero quasi periodic traveling
wave solution $u_n(t)=U(n-\nu t)$ with the above properties. In particular, for any $|\nu|>\bar R$ and $r\in(0,1)$, there is such a nonzero
quasi periodic traveling wave solution.\end{theorem}

When a nonresonance condition of Theorem \ref{th1} fails, then we have the following bifurcation results.

\begin{theorem}\label{th2}
Suppose $f\in C^2(\R_+,\R)$ with $f(0)=0$. If there are $\bar r_1\in(0,1)$, $\nu\in\RR\Phi\setminus\{0\}$ and $T>0$  such that all solutions
$k_1,k_2,\cdots,k_{m_1}\in \Z$ of equation
$$
-\nu=\Phi\left(\frac{2\pi}{T}(\bar r_1+k)\right)
$$
are either nonnegative or negative, and $m_1>0$. Then for any
$\ep>0$ small there are $m_1$ branches of nonzero quasi periodic
traveling wave solutions $u_{n,j,\ep}(t)=U_{j,\ep}(n-\nu_\ep t)$
of \eqref{e1} with $U_{j,\ep}\in C^1(\R,\C)$, $j=1,2,\cdots,m_1$,
and nonzero velocity $\nu_{\ep}$ satisfying
$U_{j,\ep}(z+T)=\eu^{2\pi \bar r_1\imath}U_{j,\ep}(z)$, $\forall
z\in \R$ along with $\nu_\ep\to \nu$ and
$U_{j,\ep}\rightrightarrows 0$ uniformly on $\R$ as $\ep\to0$.
\end{theorem}

\begin{remark}\label{remintr} If $a_j\ge 0$ for all $j\in\N$, then the assumptions of Theorem \ref{th2} are satisfies for any
$\nu\in\RR\Phi\setminus\{0\}$ such that $\frac{T}{2\pi}\Phi^{-1}(-\nu)\setminus\Z\ne\emptyset$, and so there are bifurcations of quasi periodic traveling waves in the generic resonant cases. On the other hand, if $\nu\in\RR\Phi\setminus\{0\}$ with $\frac{T}{2\pi}\Phi^{-1}(-\nu)\subset\Z$ then Theorem \ref{th1} is applicable for any $r\in(0,1)$.\end{remark}

Theorem \ref{th2} is a Lyapunov center theorem for traveling wave solutions. Similar results are derived in \cite{GLR} for
Fermi-Pasta-Ulam lattices.

We also discuss in Section \ref{higher} the extension of these results of \eqref{e1} on the lattices $\Z^2$ and $\Z^3$
\cite{CKMF,CCMK,GFB,LPK}. The final Section \ref{more} is devoted to traveling wave solutions of more general forms than above \cite{Pel1}.

\section{Existence of Traveling Wave Solutions}\label{trav}

In this section, we study the existence of traveling wave solutions of the form  $u_n(t)=U(n-\nu t)$, i.e. we are interested
in the equation
\begin{equation}\label{e2}
-\nu \imath U'(z)=\sum\limits_{j\in\N}a_j\partial_jU(z)+f(|U(z)|^2)U(z)\, ,
\end{equation}
where $z=n-\nu t$, $\nu\ne0$ and $\partial_jU(z):=U(z+j)+U(z-j)-2U(z)$. We are interested in the existence of quasi periodic solutions
$U(z)$ of \eqref{e2} stated in Theorem \ref{th1}.

\subsection{Preliminaries}\label{prelim}

In this subsection we recall some results from critical point theory of \cite {LS}. Let $H$ be a Hilbert space and let $J\in C^1(H,\R)$.
Suppose $H=H_1\oplus H_2$ for closed linear subspaces, and let $e_1,e_2,\cdots$ be the orthonormal basis of $H_1$.
Let us put $H_n^1:=\textrm{span}\, \{e_1,e_2,\cdots,e_n\}$ and $H_n:=H_n^1\oplus H_2$. Let $P_n$ be the orthogonal projection
of $H$ onto $H_n$. Set $J_n:=J/H_n$ - the restriction of functional $J$ on subspace $H_n$ -  and so $\nabla J_n(x)=P_n\nabla J(x)$ if $x\in H_n$.

\begin{definition}\label{def1}
If there are two positive constants $\al$ and $\beta$ such that
$$
\begin{gathered}
J(x)\ge0\quad \forall x\in \{x\in H_1\mid \|x\|\le \beta\}\, ,\\
J(x)\ge\al\quad \forall x\in \{x\in H_1\mid \|x\|=\beta\}\, ,\\
J(x)\le0\quad \forall x\in \{x\in H_2\mid \|x\|\le \beta\}\, ,\\
J(x)\le-\al\quad \forall x\in \{x\in H_2\mid \|x\|=\beta\}\, ,\end{gathered}
$$
then $J$ is said to satisfy the local linking condition at $0$.\end{definition}

\begin{definition}\label{def2}
We shall say that $J$ satisfies the Palais-Smale (PS)$^*$-condition if any sequence $\{x_n\}_{n\in\N}$ in $H$ such that
$x_n\in H_n$, $J(x_n)\le c<\infty$ and $P_n\nabla J(x_n)=\nabla J_n(x_n)\to0$ as $n\to \infty$, possesses a convergent subsequence.\end{definition}

Now we can state the following theorem of \cite{LS} which we apply.

\begin{theorem}\label{thma1}
Suppose
\begin{itemize}
\item[$(I_1)$]  $J\in C^1(H,\R)$ satisfies (PS)$^*$-condition.
\item[$(I_2)$] $J$ satisfies the local linking condition at $0$.
\item[$(I_3)$] $\forall n$, $J_n(x)\to-\infty$ as $\|x\|\to \infty$ and $x\in H_n$.
\item[$(I_4)$] $\nabla J=A+C$ for a bounded linear self-adjoint operator $A$ such that $AH_n\subset H_n$, $\forall n\in \N$ and
$C$ is a compact mapping.
\end{itemize}
Then $J$ possesses a critical point $\bar x$ with $|J(\bar x)|\ge \al$.\end{theorem}

\begin{remark}\label{rem1} If $0$ is an indefinite nondegenerate critical point of $J$, then $J$ satisfies the local linking
condition at $0$.\end{remark}

\subsection{Proof of Theorem \ref{th1}}

In this section, we use Theorem \ref{thma1} to prove Theorem \ref{th1}. Without loss of generality, we set $T=2\pi$.
We suppose $\nu>0$, the case $\nu<0$ can be handled similarly. First, we identify $\C$ with $\R^2$ in this section.
Let $r\in(0,1)$ be fixed. Next, we consider real Banach spaces
$$
L_r^{\wt s}:=\left\{U\in L^{\wt s}_{loc}(\R,\C)\mid U(z+2\pi)=\eu^{2\pi r\imath}U(z),
\, \forall z\in \R\right\}
$$
for $\wt s\ge 1$. Clearly $U\in L_r^{\wt s}$ if and only if $U(z)=\eu^{rz\imath}V(z)$ for some $V\in L^{\wt s}:=L^{\wt s}(S^{2\pi},\C)$.
Consequently $U_1(z+c_1)\overline{U_2(z+c_2)}$ is $2\pi$-periodic for any $c_1,c_2\in\R$ and $U_1,U_2\in L_r^{\wt s}$,
hence $|U(z)|$ is $2\pi$-periodic. So we consider the norm on
$L_r^{\wt s}$ like on $L^{\wt s}$. In particular, we have
$$
V\in L_r^{2}\Leftrightarrow V(z)=\sum\limits_{k\in \Z}V_k\eu^{(r+k)z\imath},\, V_k\in\C,\, \sum\limits_{k\in \Z}|V_k|^2<\infty\, .
$$
Let
$$
\begin{gathered}
X_r:=W_r^{1/2,2}(S^{2\pi},\C)=\left\{V\in L^{2}_r\mid V(z)=\sum\limits_{k\in \Z}V_k\eu^{(r+k)z\imath},\, \sum\limits_{k\in \Z}|V_k|^2|r+k|<\infty\right\}\, ,\\
Y_r:=W_r^{1,2}(S^{2\pi},\C)=\left\{V\in L^{2}_r\mid V(z)=\sum\limits_{k\in \Z}V_k\eu^{(r+k)z\imath},\, \sum\limits_{k\in \Z}|V_k|^2(r+k)^2<\infty\right\}\, .\end{gathered}
$$
Note $r+k\ne 0$ for any $k\in\Z$. Clearly $Y_r\subset X_r\subset L_r^2$. We consider $L_r^2$, $X_r$ and $Y_r$ as real Hilbert spaces with inner products
$$
\begin{gathered}
\langle V,W\rangle_{L_r^2}:=2\pi\Re\sum\limits_{k\in \Z}V_k\ol{W_k}=\Re \int_0^{2\pi}V(z)\ol{W(z)}dz\, ,\\
\langle V,W\rangle_{X_r}:=2\pi\Re\sum\limits_{k\in \Z}V_k\ol{W_k}|r+k|\\
\langle V,W\rangle_{Y_r}:=2\pi\Re\sum\limits_{k\in \Z}V_k\ol{W_k}(r+k)^2\end{gathered}
$$
for $V(z)=\sum\limits_{k\in \Z}V_k\eu^{(r+k)z\imath}$ and $W(z)=\sum\limits_{k\in \Z}W_k\eu^{(r+k)z\imath}$.

Clearly $\|U\|_{L^2}=\|U\|_{L_r^2}\le r_1\|U\|_{X_r}$, $\forall U\in X_r$ and  $\|U\|_{X_r}\le r_1\|U\|_{Y_r}=r_1\|U'\|_{L^2}$,
$\forall U\in Y_r$ for $r_1:=\min\left\{\sqrt{r},\sqrt{1-r}\right\}$. The following result is well-known \cite{LS,Ra1}.

\begin{lemma}\label{lem1}
For each $\wt s\ge 1$, $X_r$ is compactly embedded into $L_r^{\wt s}$.
\end{lemma}

On $X_r$, we consider a continuous symmetric bilinear form
$$
B_r(U,V):=4\pi\Re\sum\limits_{k\in \Z}V_k\ol{W_k}(r+k)\, .
$$
Note, if $U\in X_r$ and $V\in Y_r$, then
$$
2\Re \int_0^{2\pi}\imath U(z)\ol{V(z)}\, 'dz=B_r(U,V)\, .
$$
Now we consider a real functional
$$
\begin{gathered}
I_r(U):=\frac{\nu}{2}B_r(U,U)+\int_0^{2\pi}\left\{\sum\limits_{j\in\Z}\frac{a_{|j|}}{2}|U(z+j)-U(z)|^2-F(|U(z)|^2)\right\}dz\\
=\frac{\nu}{2}B_r(U,U)+\int_0^{2\pi}\left\{\sum\limits_{j\in\N}a_j|U(z+j)-U(z)|^2-F(|U(z)|^2)\right\}dz
\end{gathered}
$$
on $X_r$. Then $I_r\in C^1(X_r,\R)$ and for $U\in X_r$, $V\in Y_r$, we derive
$$
\begin{gathered}
DI_r(U)V=\\ 2\Re\left\{\int_0^{2\pi}\left(\nu\imath U(z)\ol{V(z)}\, '-\left(\sum\limits_{j\in\N}a_j\partial_jU(z)+f(|U(z)|^2)U(z)\right)\ol{V(z)}\right)dz\right\}.
\end{gathered}
$$
If $U\in X_r$ is a critical point of $I_r$ then
\begin{equation}\label{e4}
\Re\left\{\int_0^{2\pi}\left(\nu\imath U(z)\ol{V(z)}\, '-\left(\sum\limits_{j\in\N}a_j\partial_jU(z)+f(|U(z)|^2)U(z)\right)\ol{V(z)}\right)dz\right\}=0
\end{equation}
for any $V\in Y_r$. Replacing $V$ with $\imath V$ in \eqref{e4}, we obtain
$$
\int_0^{2\pi}\left(\nu\imath U(z)\ol{V(z)}\, '-\left(\sum\limits_{j\in\N}a_j\partial_jU(z)+f(|U(z)|^2)U(z)\right)\ol{V(z)}\right)dz=0
$$
for any $V\in Y_r$. This means that $U$ is a weak solution of \eqref{e2}. Then a standard regularity method shows \cite{Ra1} that $U$ is a
$C^1$-smooth solution of \eqref{e2}.

Now we split $X_r=X_+\oplus X_-$ for
$$
X_-:=\left\{V(z)=\sum\limits_{k=-\infty}^{-1}V_k\eu^{(r+k)z\imath}\right\},\quad
X_+:=\left\{V(z)=\sum\limits_{k=0}^{\infty}V_k\eu^{(r+k)z\imath}\right\}\, .
$$
Clearly if $U=U_++U_-$ then $B_r(U,U)=2\left(\|U_+\|^2_{X_r}-\|U_-\|^2_{X_r}\right)$.

Next, let us define $\wt K_r : L^2_r\to X_r$ as
\begin{equation}\label{e3}
\langle \wt K_rH,V\rangle_{X_r}:=2\Re \int_0^{2\pi}H(z)\ol{V(z)}dz,\, \forall V\in X_r\, .
\end{equation}
Then
$$
\wt K_rH=\sum_{k\in \Z}\frac{2H_k}{|r+k|}\eu^{(r+k)\imath z}
$$
and so $\wt K_r$ is compact. To study $\nabla I_r(u)$, we introduce the mapping $\Psi_r : X_r\to X_r$ defined by
$$
\langle \Psi_r(U),V\rangle _{X_r}:=2\Re \int_0^{2\pi}f(|U(z)|^2)U(z)\ol{V(z)}dz
$$
for any $V\in X_r$. By Lemma \ref{lem1}, the Nemytskij operator $U\to f(|U(z)|^2)U(z)$ from $X_r$ to $L_r^2$ is continuous.
Using \eqref{e3}, we get
$$
\Psi_r(U)=\wt K_rf(|U|^2)U\, .
$$
Hence $\Psi_r : X_r\to X_r$ is compact and continuous.

\begin{lemma}\label{lem2}
Under (H1) it hods $D\Psi_r(0)=0$.
\end{lemma}
\begin{proof} There is a constant $c_3$ such that
$$
|f(w)|\le c_3(w+w^s)
$$
for any $w\ge 0$. Then by Lemma \ref{lem1}, we derive
$$
\begin{gathered}
|f(|U|^2)U|_{L_r^2}^2=\int_0^{2\pi}f(|U(z)|^2)^2|U(z)|^2dz\\ \le 2c^2_3\int_0^{2\pi}\left(|U(z)|^6+|U(z)|^{2(2s+1)}\right)dz\le c^2_4\left(\|U\|^3_{X_r}+\|U\|^{2s+1}_{X_r}\right)^2\end{gathered}
$$
for a constant $c_4>0$. Hence
$$
\left|\langle \Psi_r(U),V\rangle_{X_r}\right|\le 2\|f(|U|^2)U\|_{L_r^2}\|V\|_{L_r^2}\le c_5\left(\|U\|^3_{X_r}+\|U\|^{2s+1}_{X_r}\right)\|V\|_{X_r}
$$
for a constant $c_5>0$. This implies
$$
\|\Psi_r(U)\|_{X_r}\le c_5\left(\|U\|^3_{X_r}+\|U\|^{2s+1}_{X_r}\right),\, \forall U\in X_r\, .
$$
Since $\Psi_r(0)=0$ and $s>0$, we get $D\Psi_r(0)=0$. The proof is finished.\end{proof}

Finally, define $\LL_r : L_r^2\to L_r^2$ as
$$
\LL_rU:=\sum_{j\in \N}a_j\partial_j U(z)\, .
$$
Then
\begin{equation}\label{ea1}
\nabla I_r(U)=\left(2\nu \I_+-2\nu \I_--\wt K_r\LL_r-\Psi_r\right)(U)
\end{equation}
for the identities $\I_\pm : X_\pm\to X_\pm$. Clearly
$$
A_r:=2\nu \I_+-2\nu \I_--\wt K_r\LL_r
$$
is a self-adjoint bounded operator $A_r : X_r\to X_r$ satisfying
$$
A_rU=2\sum_{k\in \Z}\left (\nu\, \textrm{sgn}\, (r+k)+\frac{4}{|r+k|}\sum_{j\in \N}a_j\sin^2\left[\frac{r+k}{2}j\right]\right)U_k\eu^{(r+k)\imath z}\, .
$$
Consequently, the spectrum $\sigma(A_r)$ of $A_r$ is given by
$$
\sigma(A_r)=\left \{2\, \textrm{sgn}\, (r+k)\left(\nu+\Phi(r+k)\right)\mid k\in \Z\right\}\, .
$$
By Lemma \ref{lem2}, we get that under the assumption
\begin{equation}\label{nonres}
-\nu\ne\Phi(r+k)\, \forall k\in\Z\, ,
\end{equation}
$0$ is an indefinite nondegenerate critical point of $I_r$: $\nabla I_r(0)=0$ and $\textrm{Hess}\, I_r(0)=A_r$ with $0\notin \sigma (A_r)$
and $X_r=X_{1,r}\oplus X_{2,r}$ with $\sigma (A_r/X_{1,r})\subset (0,\infty)$ and $\sigma(A_r/X_{2,r})\subset(-\infty,0)$ where
$X_{1,r}$, $X_{2,r}$ are suitable closed linear subspaces of $X_r$. Note $X_{1,r}$ and $X_{2,r}$ are infinite dimensional,
since $\Phi(r+k)\to 0$ as $|k|\to \infty$. Consequently by Remark \ref{rem1}, under \eqref{nonres}, $I_r$ satisfies the local
linking condition at $0$ in the sense of Definition \ref{def1}, i.e. condition $(I_2)$ of Theorem \ref{thma1} is verified.

We consider an equivalent scalar product $\langle\cdot,\cdot\rangle_r$ on $X_r$ such that
$$
\langle A_rU,U\rangle_r=\|U_1\|_r^2-\|U_2\|_r^2,\quad U_1\in X_{1,r},\, U_2\in X_{2,r}\, .
$$
Note there is a linear isomorphism $K_r : X_r\to X_r$ such that
$$
\langle U,V\rangle_{X_r}=\langle K_rU,V\rangle_r,\quad \forall U,\forall V\in X_{r}\, .
$$
Clearly $K_r$ is self-adjoint and positive definite. Then
$$
\begin{gathered}
I_r(U)=\frac{\nu}{2}\|U_1\|_r^2-\frac{\nu}{2}\|U_2\|_r^2-\int_0^{2\pi}F(|U(z)|^2)dz\, ,\\
\nabla I_r(U)=\nu \I_1-\nu \I_2-K_r\Psi_r\, ,\\
\langle \nabla I_r(U),V\rangle _r=DI_r(U)V=\nu\|V_1\|_r^2-\nu\|V_2\|_r^2\\ -2\Re\int_0^{2\pi}f(|U(z)|^2)U(z)\ol{V(z)}dz\, .\end{gathered}
$$
Let $X_{2,r}=\textrm{span}\, \{e_1,e_2,\cdots\}$ and $e_i$ are eigenvectors of $A_r$. Then we take
$X_n=\textrm{span}\, \{e_1,e_2,\cdots,e_n\}\oplus X_{1,r}$ for $n\ge 3$. So clearly $AX_n\subset X_n$,
i.e. condition $(I_4)$ of Theorem \ref{thma1} is verified. Let $P_n : X_r\to X_n$ be the orthogonal
projection with respect $\langle \cdot,\cdot\rangle_r$.

We suppose there is a sequence $\{U_m\}_{m\in\N}\subset X_r$, $U_m\in X_m$ and a constant $c$ such that
$$
I_r(U_m)\le c\quad\textrm{and}\quad P_m\nabla I_r(U_m)\to 0\, .
$$
Then for $m$ large we get,
\begin{equation}\label{ee1}
\begin{gathered}
c+\|U_m\|_r\ge I_r(U_m)-\frac{1}{2}\langle P_m\nabla I_r(U_m),U_m\rangle_r\\
=\int_0^{2\pi}\left[f(|U_m(z)|^2)|U_m(z)|^2-F(|U_m(z)|^2)\right]dz\\
\ge \int_0^{2\pi}(\mu-1)F(|U_m(z)|^2)dz-2\pi\bar r\\ \ge (\mu-1)c_2\int_0^{2\pi}\left(|U_m(z)|^{2(s+1)}-1\right)dz-2\pi\bar r\\
\ge (\mu-1)c_2\left(\|U_m\|^{2(s+1)}_{L^{2(s+1)}}-c_6\right)\end{gathered}
\end{equation}
for a constant $c_6>0$.

By following the same arguments, we derive
$$
\begin{gathered}
\nu \|U_{1,m}\|^2_r\le \|P_m\nabla I_m(U_m)\|\cdot\|U_{1,m}\|_r+2\int_0^{2\pi}f(|U_m(z)|^2)|U_m(z)||U_{1,m}(z)|dz\\
\le\|U_{1,m}\|_r+2c_7\int_0^{2\pi}\left(|U_m(z)|^{2s+1}+1\right)|U_{1,m}(z)|dz\\
\le\|U_{1,m}\|_r+2c_7\left\| |U_m|^{2s+1}+1\right\|_{L^{\frac{2(s+1)}{2s+1}}}\|U_{1,m}\|_{L^{2(s+1)}}\\
\le\|U_{1,m}\|_r+2c_7\left(\|U_m\|^{2s+1}_{L^{2(s+1)}}+1\right)\|U_{1,m}\|_r
\end{gathered}
$$
and hence
$$
\|U_{1,m}\|_r\le c_8\left(\|U_m\|^{2s+1}_{L^{2(s+1)}}+1\right)\, .
$$
Similarly we obtain
$$
\|U_{2,m}\|_r\le c_8\left(\|U_m\|^{2s+1}_{L^{2(s+1)}}+1\right)
$$
and consequently by \eqref{ee1}, we obtain
$$
\|U_{m}\|_r\le 2c_{8}\left(\|U_m\|^{2s+1}_{L^{2(s+1)}}+1\right)\le c_{9}\left(\|U_m\|^{\frac{2s+1}{2(s+1)}}_r+1\right)
$$
for positive constants $c_7$, $c_8$ and $c_{9}$. Thus $\{U_m\}_{m\in\N}\subset X_r$ is bounded. Since
$$
P_m\nabla I_r(U_m)=\nu U_{1,m}-\nu U_{2,m}-K_r\Psi_r(U_m)\to 0
$$
and $K_r\Psi_r$ is compact, there is a convergent subsequence of $\{U_m\}_{m\in\N}$ in $X_r$. Summarizing, (PS)$^*$-condition is verified for $I_{r}$, i.e. condition $(I_1)$ of Theorem \ref{thma1} is verified.

Next, let $U\in X_n$. Then using $U_1\in \textrm{span}\, \{e_1,e_2,\cdots,e_n\}$, we derive
$$
\begin{gathered}
I_r(U)=\frac{\nu}{2}\left(\|U_1\|_r^2-\|U_2\|_r^2\right)-\int_0^{2\pi}F(|U(z)|^2)dz\\
\le\frac{\nu}{2}\left(\|U_1\|_r^2-\|U_2\|_r^2\right)-c_{2}\int_0^{2\pi}\left(|U(z)|^{2(s+1)}-1\right)dz\\
\le\frac{\nu}{2}\left(\|U_1\|_r^2-\|U_2\|_r^2\right)-c_{2}\|U\|^{2(s+1)}_{L^{2(s+1)}}+c_{10}\\
\le\frac{\nu}{2}\left(\|U_1\|_r^2-\|U_2\|_r^2\right)-c_{11}\left(\|U_1\|^{2(s+1)}_{L^{2}}+\|U_2\|^{2(s+1)}_{L^{2}}\right)+c_{10}\\
\le \frac{\nu}{2}\|U_1\|_r^2\left(1-c_{12}\|U_1\|^{2s}_r\right)-\frac{\nu}{2}\|U_2\|_r^2+c_{10}
\end{gathered}
$$
for positive constants $c_{10}$, $c_{11}$ and $c_{12}$. Now it is clear that $I_r(U)\to -\infty$ as $\|U\|_r\to\infty$,
i.e. condition $(I_3)$ of Theorem \ref{thma1} is verified.

Summarizing, under assumptions (H1) and \eqref{nonres}, all conditions $(I_1)$-$(I_4)$ of Theorem \ref{thma1}
are verified for $I_r$. Hence there is a nonzero critical point $U_r\in X_r$ of $I_r$, which we already know to be a
$C^1$-smooth solution of \eqref{e2} satisfying \eqref{prop}. Note \eqref{nonres} certainly holds for any $|\nu|>\bar R$ and
$r\in(0,1)$. Hence the proof of the second part of Theorem \ref{th1} is finished. To prove the first part, it enough to observe that the set
$$
\left\{\Phi(r+k)\mid r\in \Q\cap (0,1), \quad k\in\Z\right\}
$$
is countable, and thus for almost each $\nu\in\R\setminus\{0\}$ and any $r\in \Q\cap (0,1)$, condition \eqref{nonres} holds.

\subsection{Remarks}\label{remarks}

\begin{remark}\label{rem2} When $r$ is rational in Theorem \ref{th1} then we get periodic $U(z)$ with arbitrarily large
minimal periods. If $r$ is irrational then clearly $U(z)=\eu^{\frac{2\pi}{T}rz\imath}V(z)$ for a
$T$-periodic $V(z)=U(z)\eu^{-\frac{2\pi}{T}rz\imath}$. So $U(z)$ is quasi periodic and its orbit in $\C\simeq \R^2$ is dense either
in a compact annulus or in a compact disc. But $|U(z)|$ is $T$-periodic in the both cases.\end{remark}

\begin{remark}\label{rem3} Changing $t\leftrightarrow -t$, we can also handle DNLS
\begin{equation}\label{re1}
-\imath \dot u_n=\sum\limits_{j\in\N}a_j\Delta_ju_{n}+f(|u_n|^2)u_n,\quad n\in\Z
\end{equation}
under (H1) and \eqref{nonres} becomes
\begin{equation}\label{rnonres}
\nu\ne\Phi(r+k)\, \forall k\in\Z\quad \textrm{and}\quad \nu\in(0,\bar R)\, .
\end{equation}
\end{remark}

\begin{remark}\label{rem4}
Assume that $U\in Y_r$ is a weak solution of \eqref{e2}, then
$$
\begin{gathered}
|U(z)|\le\sum_{k\in \Z}|U_k|\le\sqrt{\sum_{k\in \Z}|U_k|^2(r+k)^2}\sqrt{\sum_{k\in \Z}(r+k)^{-2}}\\ =\sqrt{\frac{\pi}{2}}{\rm cosec}\, \pi r\, \|U\|_{Y_r}.\end{gathered}
$$
Let $\wt R:=\max_{x\in \R_+}x\Phi(x)$. Then
$$
\begin{gathered}
|\nu|\|U'\|_{L^2}=|\nu|\|U\|_{Y_r}\le\|\LL_rU\|_{L^2}+\|f(|U|^2U)\|_{L^2}\\ \le \wt R\|U\|_{L^2}+c_1\left\||U|^{2s+1}+|U|\right\|_{L^2}\\ \le \left(\wt Rr_1^2+c_1r_1^2+c_1\frac{\pi^{s+1}}{2^s}{\rm cosec}^{2s+1}\, \pi r\, \|U\|_{Y_r}^{2s}\right)\|U\|_{Y_r}.\end{gathered}
$$
So if $U\ne 0$ then we obtain
$$
|\nu|\le\wt Rr_1^2+c_1r_1^2+c_1\frac{\pi^{s+1}}{2^s}{\rm cosec}^{2s+1}\, \pi r\, \|U\|_{Y_r}^{2s},
$$
i.e.
$$
\|U\|_{Y_r}\ge\sqrt{2}\sqrt[2s]{\frac{|\nu|-\wt Rr_1^2-c_1r_1^2}{c_1\pi^{s+1}{\rm cosec}^{2s+1}\, \pi r}}
$$
for
$$
|\nu|\ge \wt Rr_1^2+c_1r_1^2.
$$
Hence $\|U\|_{Y_r}\to \infty$ as $|\nu|\to \infty$ for a possible nonzero solution $U\in Y_r$ of \eqref{e2}.
\end{remark}

\subsection{Examples}\label{examples}

We first note
\begin{equation}\label{ex1}
\Phi(x)=\frac{2}{x}\sum_{j\in\N}a_j(1-\cos x j)=\frac{2}{x}\left[\sum_{j\in\N}a_j-\Re\sum_{j\in\N}a_j\eu^{xj\imath}\right].
\end{equation}
Now we turn the the following concrete examples.

\begin{example}\label{exam1}
First we suppose that $a_j$ is decaying rapidly to $0$. Let $a_j=\frac{1}{j!}$. Then
$$
\begin{gathered}
\sum_{j\in\N}\frac{1}{j!}\eu^{xj\imath}=\eu^{\eu^{x\imath}}-1=\eu^{\cos x +\imath\sin x}-1\\ =\eu^{\cos x}\left[\cos\sin x+\imath\sin \sin x\right]-1.
\end{gathered}
$$
So by \eqref{ex1} we derive
$$
\Phi(x)=\frac{2}{x}\left[\sum_{j\in\N}\frac{1}{j!}-\eu^{\cos x}\cos\sin x+1\right]=\frac{2}{x}\left[\eu-\eu^{\cos x}\cos\sin x\right].
$$
By Remark \ref{inrem1}, $\Phi\in C^1(\R,\R)$ with the graph on $[-4\pi,4\pi]$:
\begin{center}
\includegraphics[clip,width=6cm]{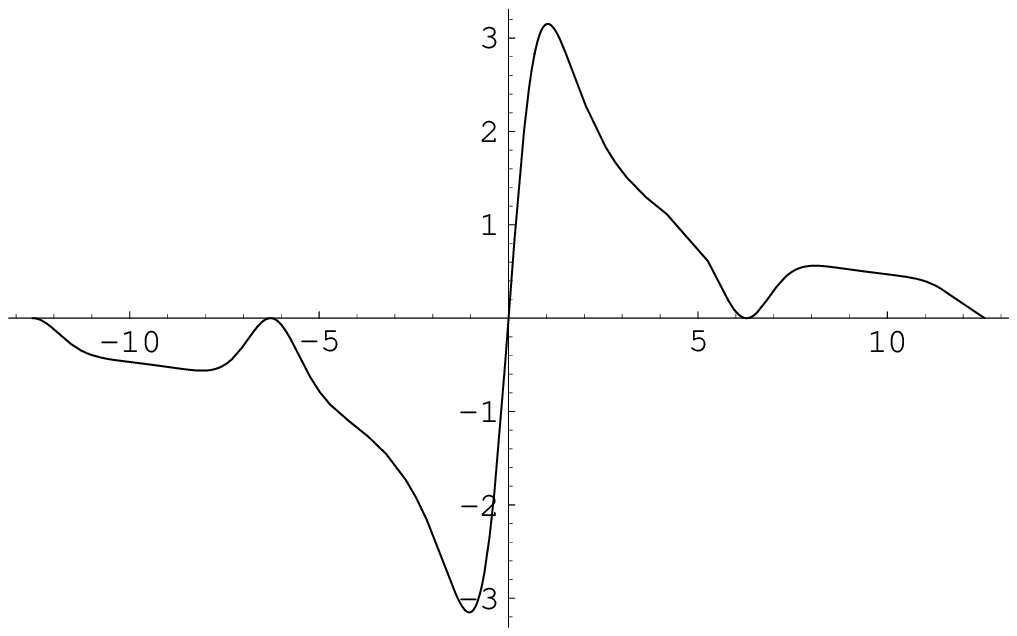}
\end{center}
A numerical solution shows that $\Phi$ has a maximum $\bar R=\Phi(x_0)\doteq3.15177$ at
$x_0\doteq1.03665$.
\end{example}

\begin{example}\label{exam2}
Now we suppose that $a_j$ is decaying exponentially to $0$. Let $a_j=\eu^{-j}$, hence we have the discrete
Kac-Baker interaction kernel \cite{CGGJMR,CCK}. Then
$$
\begin{gathered}
\sum_{j\in\N}\eu^{-j}\eu^{xj\imath}=\sum_{j\in\N}\eu^{(x\imath-1)j}=\frac{\eu^{x\imath-1}}{1-\eu^{x\imath-1}}\\ =\frac{\cos x+\imath\sin x}{\eu-\cos x-\imath \sin x}=\frac{\eu\cos x-1+\eu\imath\sin x}{\eu^2+1-2\eu\cos x}.
\end{gathered}
$$
So by \eqref{ex1} we derive
$$
\Phi(x)=\frac{2}{x}\left[\sum_{j\in\N}\eu^{-j}-\frac{\eu\cos x-1}{\eu^2+1-2\eu\cos x}\right]=\frac{2\eu(\eu+1)(1-\cos x)}{(\eu-1)x(\eu^2+1-2\eu\cos x)}.
$$
By Remark \ref{inrem1}, $\Phi\in C^1(\R,\R)$ with the graph on $[-4\pi,4\pi]$:
\begin{center}
\includegraphics[clip,width=6cm]{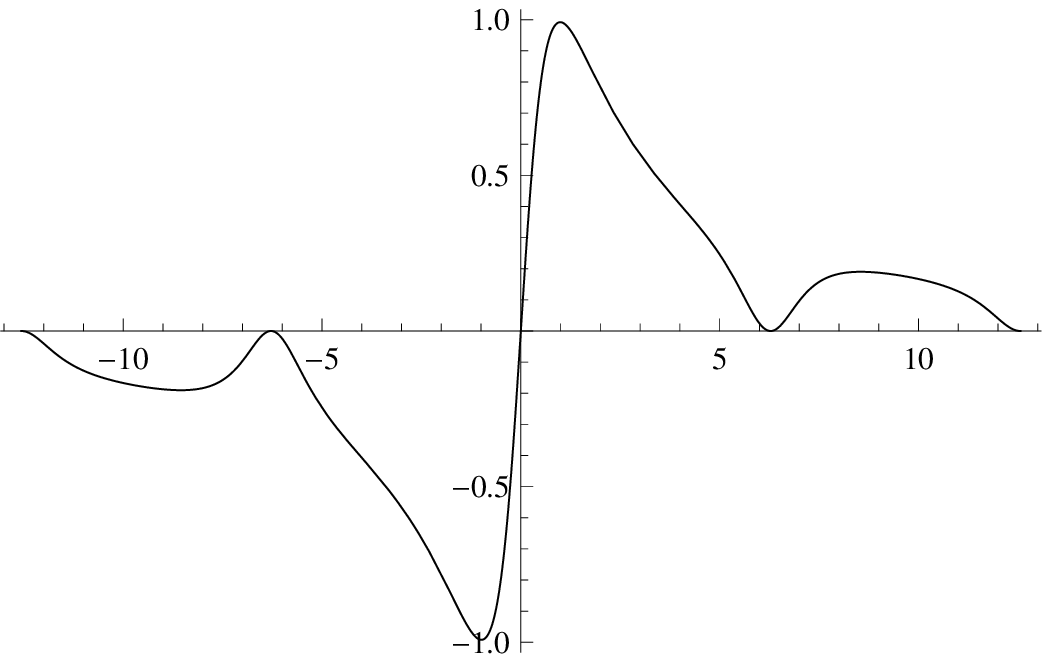}
\end{center}
A numerical solution shows that $\Phi$ has a maximum $\bar R=\Phi(x_0)\doteq0.992045$ at $x_0\doteq0.991541$.\end{example}

\begin{example}\label{exam3}
In this example, we suppose that $a_j$ is decaying polynomially to $0$ (cf. \cite{MG}), by considering several cases:

\

1. Let $a_j=\frac{1}{j^4}$. Then
$$
\Phi(x)=\frac{2}{x}\sum_{j\in\N}\left(\frac{1}{j^4}-\frac{1}{j^4}\cos xj\right)=\frac{\left(|x|-2\pi\left[\frac{|x|}{2\pi}\right]\right)^2}{24x}\left(2\pi-|x|+2\pi\left[\frac{|x|}{2\pi}\right]\right)^2.
$$
Here $[\cdot]$ is the integer part function. By Remark \ref{inrem1}, $\Phi\in C^1(\R,\R)$ with the graph on $[-4\pi,4\pi]$:
\begin{center}
\includegraphics[clip,width=6cm]{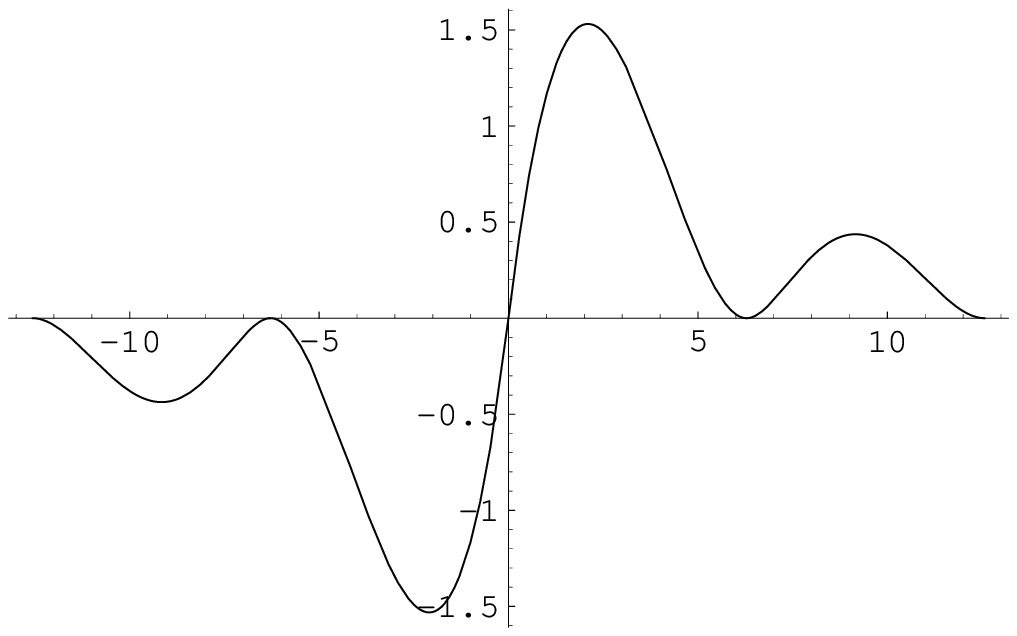}
\end{center}
$\Phi$ has a maximum $\bar R=\Phi(x_0)=\frac{4\pi^3}{81}\doteq 1.53117$ at $x_0=2\pi/3\doteq 2.0944$.
Similar results hold for $a_j=j^{-\beta}$ with $\beta >3$ by Remark \ref{inrem1}.

\

2. Let $a_j=\frac{1}{j^3}$. So we consider the dipole-dipole interaction (cf. \cite{ACG,CGGJMR,LCh,MG}).
By Remark \ref{inrem1}, $\Phi\in C(\R,\R)$ with the graph on $[-4\pi,4\pi]$:
\begin{center}
\includegraphics[clip,width=6cm]{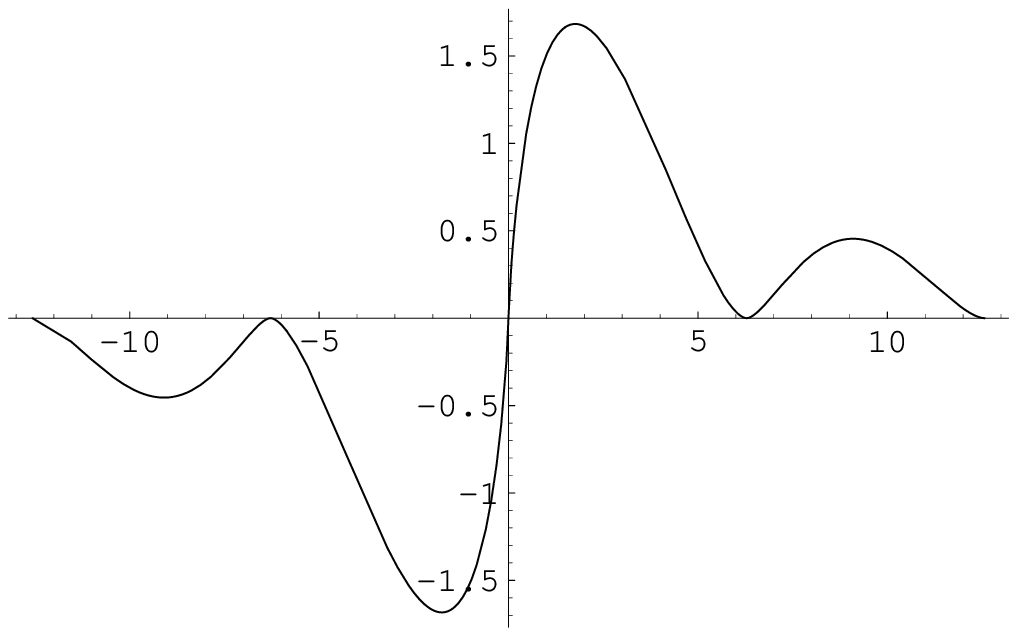}
\end{center}
$\Phi$ has a maximum $\bar R=\Phi(x_0)\doteq 1.68311$ at $x_0\doteq 1.76076$. Next we know that \cite{Zyg}
$$
\sum_{j\in\N}\frac{1}{j}\cos xj=-\ln\left|2\sin\frac{x}{2}\right|,\quad 0<x<2\pi.
$$
Then
$$
\sum_{j\in\N}\frac{1}{j^2}\sin xj=-\int\limits_0^x\ln\left|2\sin\frac{s}{2}\right|ds.
$$
Using $x/2\le \sin x\le x$ for $x\ge 0$ small, we derive
$$
x-x\ln x=-\int\limits_0^x\ln s \, ds\le \sum_{j\in\N}\frac{1}{j^2}\sin xj\le-\int\limits_0^x\ln\frac{s}{2}ds=x-x\ln\frac{x}{2}.
$$
By L'Hopital's rule, we obtain
$$
\lim_{x\to 0_+}\frac{\Phi(x)}{x}=\lim_{x\to 0_+}\frac{4\sum_{j\in\N}\frac{1}{j^3}\sin^2 xj}{x^2}=\lim_{x\to 0_+}\frac{2\sum_{j\in\N}\frac{1}{j^2}\sin 2xj}{x}=+\infty.
$$
Hence $\Phi$ has no derivative at $x_0=0$.

Next, let $a_j=j^{-\beta}$ for $2<\beta<3$. By Remark \ref{inrem1}, $\Phi$ is still continuous. Since $\Phi(0)=0$ and
$$
\lim_{x\to 0_+}\frac{\Phi(x)}{x}\ge\lim_{x\to 0_+}\frac{4\sum_{j\in\N}\frac{1}{j^3}\sin^2 xj}{x^2}=+\infty,
$$
$\Phi(x)$ is continuous but not $C^1$-smooth on $\R$.

\

3. Let $a_j=\frac{1}{j^2}$. Then
$$
\Phi(x)=\frac{2}{x}\sum_{j\in\N}\left(\frac{1}{j^2}-\frac{1}{j^2}\cos xj\right)=\frac{\left(|x|-2\pi\left[\frac{|x|}{2\pi}\right]\right)}{2x}\left(2\pi-|x|+2\pi\left[\frac{|x|}{2\pi}\right]\right).
$$
By Remark \ref{inrem1}, $\Phi\in C(\R\setminus\{0\},\R)$ with the graph on $[-4\pi,4\pi]$:
\begin{center}
\includegraphics[clip,width=6cm]{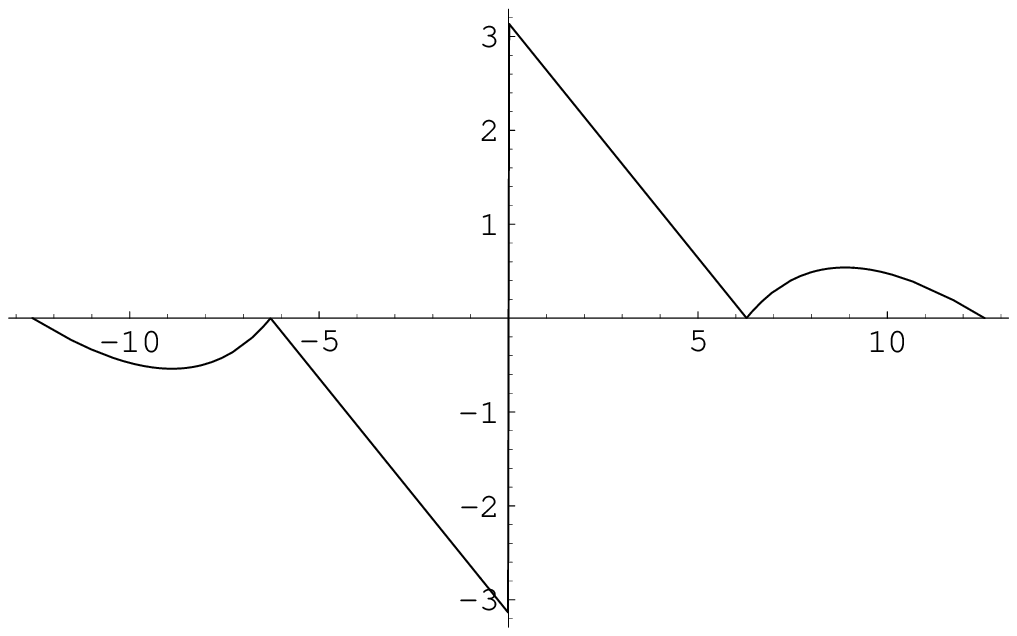}
\end{center}
$\Phi$ is discontinuous at $x_0=0$ where it has a supremum $\bar R=\pi$.

\

4. Let $a_j=j^{-\beta}$ for $1<\beta<2$. For $\beta =7/4$, $\Phi$ has the graph on $[-4\pi,4\pi]$:
\begin{center}
    \includegraphics[clip,width=6cm]{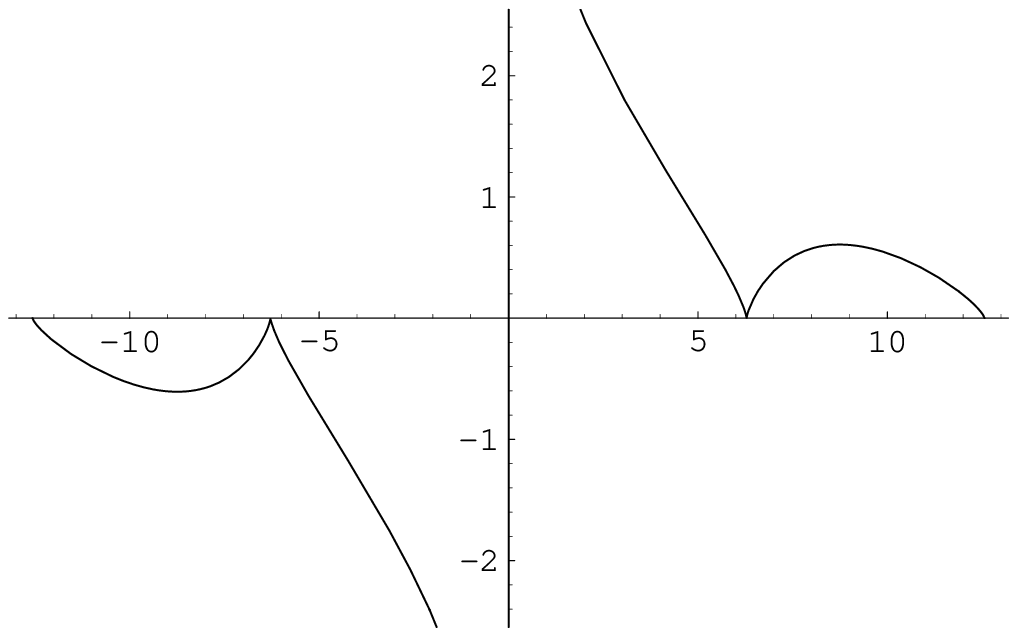}
\end{center}
Hence $\Phi$ is discontinuous at $x_0=0$ with $\lim_{x\to0_+}\Phi(x)=+\infty$. We show that this holds for any $1<\beta<2$. First suppose $3/2<\be<2$. Then the series
$$
\Upsilon(x):=\sum_{j\in\N}\frac{1}{j^{\be-1}}\sin jx
$$
converges uniformly on any $[\ep,2\pi-\ep]$ for $0<\ep<\pi$. But since $\sum_{j\in\N}\frac{1}{j^{2(\be-1)}}<\infty$, so $\Upsilon\in L^2\subset L^1$. On the other hand, we know \cite{Zyg} that
$$
\Upsilon(x):=\Gamma(2-\be)\cos\frac{\pi(\be-1)}{2}\cdot x^{\be-2}+O(1)
$$
on $(0,\pi]$. Hence
$$
\sum_{j\in\N}\frac{1-\cos jx}{j^{\be}}=\int_0^x\Upsilon(s)ds=\frac{\Gamma(2-\be)}{\be-1}\cos\frac{\pi(\be-1)}{2}\cdot x^{\be-1}+O(x)
$$
on $[0,\pi]$. Consequently, we obtain
$$
\Phi(x)=\frac{2\Gamma(2-\be)}{\be-1}\cos\frac{\pi(\be-1)}{2}\cdot x^{\be-2}+O(1)
$$
on $(0,\pi]$, which implies $\lim_{x\to0_+}\Phi(x)=+\infty$ for any $3/2<\be<2$. Finally, if $1<\be\le 3/2$, then
$$
\Phi(x)\ge \frac{2}{x}\sum_{j\in\N}\frac{1-\cos jx}{j^{7/4}}=\frac{8}{3}\Gamma\left(\frac{1}{4}\right)\cos\frac{3\pi}{8}\cdot \frac{1}{\sqrt[4]{x}}+O(1)\to+\infty
$$
as $x\to0_+$. Hence, $\lim_{x\to0_+}\Phi(x)=+\infty$ for any $1<\be<2$.\end{example}

Summarizing, we have the following result.

\begin{lemma}\label{exlem1} Let $a_j=j^{-\be}$ for $1<\be$. Then
\begin{itemize}
\item[(i)] $\Phi\in C^1(\R,\R)$ for $\be>3$, and $\RR\Phi=[-\bar R,\bar R]$ for some $\bar R<\infty$.
\item[(ii)] $\Phi\in C(\R,\R)$ and  $\Phi\notin C^1(\R,\R)$ for $2<\be\le 3$, and $\RR\Phi=[-\bar R,\bar R]$ for some $\bar R<\infty$.
\item[(iii)] $\Phi\in C(\R\setminus\{0\},\R)$ and  $\Phi\notin C(\R,\R)$ for $\be=2$, and $\RR\Phi=(-\pi,\pi)$.
\item[(iv)] $\Phi\in C(\R\setminus\{0\},\R)$ and  $\Phi\notin C(\R,\R)$ for $1<\be<2$, and $\RR\Phi=(-\infty,+\infty)$.
\end{itemize}
\end{lemma}

\begin{remark}\label{exrem1} We see that if the interaction is strong, so the case (iv) of Lemma \ref{exlem1} holds,
then there are continuum many quasi periodic traveling wave solutions $U(z)$ of Theorem \ref{th1} for any $\nu\ne0$, $T>0$ and $r\in(0,1)$
such that $r\notin\left\{z-[z]\mid z\in \frac{T}{2\pi}\Phi^{-1}(-\nu)\right\}$, with $\|U\|_{Y_r}\to\infty$ as $|\nu|\to\infty$ by Remark \ref{rem4}.
On the other hand, if the interaction is weak, then we can show in addition quasi periodic traveling waves with speeds in intervals $(-\infty,-\bar R)$
and $(\bar R,\infty)$ for any $T>0$ and $r\in(0,1)$.\end{remark}

\begin{remark}\label{exremm1a} For the reader convenience, we present the above graphs of function $\Phi$ to visualize
their quantitative and qualitative changes according to different choices of values of sequences $\{a_j\}_{j\in\Z}$ in
\eqref{e1}, and hence with different consequences from Theorems \ref{th1} and \ref{th2} for the existence and bifurcations of
quasi periodic traveling wave solutions of \eqref{e1}. Moreover, these graphs can be compared with similar ones for traveling
waves for higher dimensional DNLS in Section \ref{higher} and for traveling waves with frequencies in Section \ref{more}.
Finally these examples are motivated by applications mentioned in the corresponding references.\end{remark}

\section{Bifurcation of Traveling Wave Solutions}\label{bifur}

In this section we proceed with the study of \eqref{e2} when nonresonance of Theorem \ref{th1} fails, i.e. $r\in\{\bar r_1,\bar r_2,\cdots,\bar r_m\}$.
We scale in \eqref{e2} the velocity by $\nu\leftrightarrow \nu/(1+\la)$ to get equation
\begin{equation}\label{be2}
-\nu \imath U'(z)=(1+\la)\left(\sum\limits_{j\in\N}a_j\partial_jU(z)+f(|U(z)|^2)U(z)\right)\, ,
\end{equation}
where $\la$ is a small parameter, i.e. $u_n(t)=U\left(n-\frac{\nu}{1+\la}t\right)$ is a solution of \eqref{e1}. We are interested in the existence
of quasi periodic solutions $U(z)$ of \eqref{be2} stated in Theorem \ref{th2}.

\subsection{Preliminaries}\label{prelim2}

In this subsection we recall some results from critical point theory of \cite {M1}. Let $H$ be a Hilbert space with a scalar product $(\cdot,\cdot)$
and the corresponding norm $\|\cdot\|$. Let $\Theta : S^1\to L(H)$ be an isometric representation of the unit circle $S^1$ over $H$,
i.e. the following properties are satisfied
\begin{itemize}
\item[(R)] $\Theta(0)=\I$ - the identity, $\Theta(\theta_1+\theta_2)=\Theta(\theta_1)\Theta(\theta_2)$ for any $\theta_1,\theta_2\in S^1$, $(\theta,h)\to \Theta(\theta)h$ is continuous, and $\|\Theta(\theta)h\|=\|h\|$ for any $\theta\in S^1$ and $h\in H$.
\end{itemize}
We set
$$
\textrm{Fix}(S^1):=\left\{h\in H\mid \Theta(\theta)h=h\, \forall \theta\in\Theta\right\}.
$$
We consider $J_1,J_2\in C^2(H,\R)$ such that
\begin{itemize}
\item[(H1)] $J_2(0)=0$ and $\nabla J_1(0)=\nabla J_2(0)=0$.
\item[(H2)] $\textrm{Hess}\, J_1(0)$ is a Fredholm operator, i.e. $\dim \textrm{Hess}\, J_1(0)<\infty$, $\RR \textrm{Hess}\, J_1(0)$ is closed and $\textrm{codim}\, \RR \textrm{Hess}\, J_1(0)<\infty$.
\item[(H3)] $\dim \ker \textrm{Hess}\, J_1(0)\ge 2$ and $\textrm{Hess}\, J_2(0)$ is positive definite on $\ker \textrm{Hess}\, J_1(0)$.
\item[(H4)] $J_1$ and $J_2$ are $S^1$-invariant, i.e. $J_{1,2}(\Theta(\theta)h)=\Theta(\theta)J_{1,2}(h)$ for any $\theta\in\Theta$ and $h\in H$.
\item[(H5)] $\ker \textrm{Hess}\, J_1(0)\cap \textrm{Fix}(S^1)=\{0\}$.
\end{itemize}
Now we can state the following \cite[Theorem 6.7]{M1}.

\begin{theorem}\label{bifmaw}
Under the above assumptions (H1)-(H5), for each sufficiently small $\ep>0$, equation
\begin{equation}\label{ebifmaw}
\nabla J_1(h)+\la\nabla J_2(h)=0
\end{equation}
has at leat $\frac{1}{2}\dim \ker \textrm{\rm Hess}\, J_1(0)$ of $S^1$-orbit solutions
$$
\left\{(\la_k(\ep),\Theta(\theta))h_k(\ep)\mid \theta \in S^1\right\},\quad k=1,2,\cdots, \frac{1}{2}\dim \ker \textrm{\rm Hess}\, J_1(0)
$$
such that $J_2(h_k(\ep))=\ep$ and $h_k(\ep)\to 0$, $\la_k(\ep)\to 0$ as $\ep\to 0$. Clearly $h_k(\ep)\ne 0$.\end{theorem}

\begin{remark}\label{bifrem1}
When $\textrm{Hess}\, J_2(0)$ is negative definite on $\ker \textrm{Hess}\, J_1(0)$, then Theorem \ref{bifmaw} holds for $\ep<0$ small.\end{remark}

\begin{remark}\label{birem1b} By (H4), $\ker \textrm{Hess}\, J_1(0)$ is invariant with respect to $\Theta$. Using (H5),
$\dim \ker \textrm{Hess}\, J_1(0)$ is even.\end{remark}

Now assume $H=H_+\oplus H_-$ be an orthogonal and $\Theta$-invariant decomposition with the corresponding orthogonal
projections $P_\pm : H\to H_\pm$. Then $\Theta(\th)P_\pm=P_\pm\Theta(\th)$ for any $\th\in\Theta$. Let us consider an equation
\begin{equation}\label{bifeq2}
\zeta(\I_+ -\I_-)h+(1+\la)(\KK h+\nabla\FF(h))=0,
\end{equation}
where $\zeta \ne0$ is a constant, $\la$ is a small parameter, $\I_\pm : H_\pm\to H_\pm$ are the identities. We suppose
\begin{itemize}
\item[(A)] $\KK : H\to H$ is compact self-adjoint and $\FF\in C^2(H,\R)$ with $\FF(0)=0$, $\nabla\FF(0)=0$, $\textrm{\rm Hess}\, \FF(0)=0$, and
$\KK$, $\FF$ are $S^1$-invariant. Moreover, $\KK H_\pm\subset H_\pm$.
\end{itemize}
Then
$$
\begin{gathered}
J_1(h)=\frac{\zeta}{2}(\|P_+h\|^2 -\|P_-h\|^2)+\frac{1}{2}(\KK h,h)+\FF(h),\\
J_2(h)=\frac{1}{2}(\KK h,h)+\FF(h).\end{gathered}
$$
Hence
$$
\begin{gathered}
J_1(0)=J_2(0)=0,\quad \nabla J_1(0)=\nabla J_2(0)=0,\\
\textrm{Hess}\, J_1(0)=\zeta(\I_+ -\I_-)+\KK,\quad \textrm{Hess}\, J_2(0)=\KK.\end{gathered}
$$
So assumptions (H1), (H2) and (H4) are satisfied. Since $P_\pm\KK=\KK P_\pm$, equation
$$
\textrm{Hess}\, J_1(0)h=\zeta(\I_+ -\I_-)h+\KK h=0
$$
splits into
$$
\KK h_+=-\zeta h_+,\quad \KK h_-=\zeta h_-,\quad h_\pm=P_\pm h.
$$
Consequently, supposing either
\begin{itemize}
\item[(B$_+$)] $\ker (\zeta \I+\KK)\cap H_+=\{0\}$, $\dim \ker (\zeta \I-\KK)\cap H_-\ge 2$ and $\ker (\zeta \I-\KK)\cap H_-\cap\textrm{Fix}(S^1)=\{0\}$
\end{itemize}
or
\begin{itemize}
\item[(B$_-$)] $\ker (\zeta \I-\KK)\cap H_-=\{0\}$, $\dim \ker (\zeta \I+\KK)\cap H_+\ge 2$ and $\ker (\zeta \I+\KK)\cap H_+\cap\textrm{Fix}(S^1)=\{0\}$
\end{itemize}
we get either
$$
\ker\textrm{Hess}\, J_1(0)=\ker (\zeta \I-\KK)\cap H_-
$$
or
$$
\ker\textrm{Hess}\, J_1(0)=\ker (\zeta \I+\KK)\cap H_+
$$
and so (H5) holds as well. Finally, we derive
$$
\textrm{Hess}\, J_2(0)|\ker\textrm{Hess}\, J_1(0)=\pm \zeta \I
$$
and thus (H3) is also verified (cf. Remark \ref{bifrem1}). Summarizing, Theorem \ref{bifmaw} and Remark \ref{bifrem1} is applicable to \eqref{bifeq2}:

\begin{corollary}\label{bifcor1}
Under assumptions (A) and (B$_\pm$), for each sufficiently small $\ep\ne0$, $\pm \ep\zeta>0$, equation \eqref{bifeq2}
has at leat $\frac{1}{2}\dim \ker (\zeta \I\mp\KK)\cap H_\mp$ of $S^1$-orbit solutions
$$
\left\{(\la_k(\ep),\Theta(\theta))h_k(\ep)\mid \theta \in S^1\right\},\quad k=1,2,\cdots, \frac{1}{2}\dim \ker (\zeta \I\mp\KK)\cap H_\mp
$$
such that $\frac{1}{2}(\KK h_k(\ep),h_k(\ep))+\FF(h_k(\ep))=\ep$ and $h_k(\ep)\to 0$, $\la_k(\ep)\to 0$ as $\ep\to 0$. Clearly $h_k(\ep)\ne 0$.\end{corollary}

\begin{remark}\label{bifrem2}
If $\textrm{Fix}(S^1)=\{0\}$ then (B$_+$) holds if
\begin{itemize}
\item[(i)] $-\zeta\notin \sigma(\KK/H_+)$, $\zeta\in \sigma(\KK/H_-)$ and $\zeta$ has a multiplicity at least $2$,
\end{itemize}
while (B$_-$) holds if
\begin{itemize}
\item[(ii)] $-\zeta\in \sigma(\KK/H_+)$, $\zeta\notin \sigma(\KK/H_-)$ and $-\zeta$ has a multiplicity at least $2$,
\end{itemize}
respectively.
\end{remark}

\subsection{Proof of Theorem \ref{th2}}

We again assume for simplicity $T=2\pi$. So let $r=\bar r_1\in(0,1)$ and the equation
$$
-\nu=\Phi\left(\bar r_1+k\right)
$$
has solutions $k_1,k_2,\cdots,k_{m_1}\in\Z$ which are either all nonnegative, or all negative. Next \eqref{be2} has the form (cf. \eqref{ea1})
\begin{equation}\label{be3}
2(\nu \I_+-\nu \I_-)-(1+\la)\left(\wt K_r\LL_rU+\Psi_r(U)\right)=0
\end{equation}
and
$$
\begin{gathered}
H=X_r,\quad \zeta=2\nu,\quad H_\pm=X_\pm,\\
\KK=-\wt K_r\LL_r,\quad \FF(u)=-\int_0^{2\pi}F(|U(z)|^2)dz.\end{gathered}
$$
Isometric representation $\Theta$ is naturally given as
$$
\Theta(\th)U(z):=U(z+\th),
$$
i.e.
$$
\Theta(\th)\left(\sum_{k\in\Z}U_k\eu^{(\bar r_1+k)z\imath}\right)=\sum_{k\in\Z}U_k\eu^{\th k\imath}\eu^{(\bar r_1+k)z\imath}.
$$
Note $\textrm{Fix}(S^1)=\{0\}$. It is easy to verify (R) for $\Theta$. By results of Section \ref{trav}, we get both $\KK H_\pm\subset H_\pm$ and assumption (A) holds, and moreover
$$
\sigma\left(\KK/H_\pm\right)=\left\{\pm 2\Phi(\bar r_1+k)\mid k\in \Z_\pm\right\}.
$$
Note $\Z_+=\{0\}\cup\N$ and $\Z_-=-\N$. Hence (i) of Remark \ref{bifrem2} is satisfied if
$$
-\nu\notin\left\{\Phi(\bar r_1+k)\mid k\in\Z_+\right\},\quad -\nu\in\left\{\Phi(\bar r_1+k)\mid k\in\Z_-\right\},
$$
while (ii) if
$$
-\nu\in\left\{\Phi(\bar r_1+k)\mid k\in\Z_+\right\},\quad -\nu\notin\left\{\Phi(\bar r_1+k)\mid k\in\Z_-\right\}.
$$
But these are precisely assumptions of Theorem \ref{th2}. So its proof is complete by Corollary \ref{bifcor1} and Remark \ref{bifrem2}.

\section{Traveling Waves for Higher Dimensional DNLS}\label{higher}

In this section, we first show how to extend previous results for 2-dimensional DNLS (2D DNLS) \cite{CKMF,CCMK,GFB} of forms
\begin{equation}\label{2e1}
\begin{gathered}
\imath \dot u_{n,m}=\sum\limits_{(i,j)\in\Z^2_0}a_{i,j}\Delta_{i,j}u_{n,m}+f(|u_{n,m}|^2)u_{n,m},\quad (n,m)\in\Z^2\\
=2\sum\limits_{(i,j)\in\Z^2_0}a_{i,j}\left(u_{n+i,m+j}-u_{n,m}\right)+f(|u_{n,m}|^2)u_{n,m},
\end{gathered}
\end{equation}
where $u_{n,m}\in\C$, $\Z_0^2:=\Z^2\setminus\{(0,0)\}$, $\Delta_{i,j}u_{n,m}:=u_{n+i,m+j}+u_{n-i,m-j}-2u_{n,m}$ are $2$-dimensional discrete Laplacians, $f$ satisfies (H1) and $a_{i,j}=a_{-i,-j}$ along with $\sum\limits_{(i,j)\in\Z^2_0}|a_{i,j}|<\infty$ and all $a_{i,j}$ are not zero.

Again, \eqref{2e1} conserves two dynamical invariants
$$
\begin{gathered}
\sum\limits_{(n,m)\in\Z^2}|u_{n,m}|^2\quad -\textrm{the norm},\\
\sum\limits_{(n,m)\in\Z^2}\left[-\sum\limits_{(i,j)\in\Z^2_0}a_{i,j}\left|u_{n+i,m+j}-u_{n,m}\right|^2+F(|u_{n,m}|^2)\right]\quad -\textrm{the energy}.
\end{gathered}
$$
We look for traveling wave solutions of \eqref{2e1} of the form
\begin{equation}\label{2e2a}
u_{n,m}(t)=U(n\cos \th+m\sin\th-\nu t)
\end{equation}
with a direction $(\cos \th,\sin\th)$ \cite{FR2}.  Hence we are interested in the equation
\begin{equation}\label{2e2}
-\nu \imath U'(z)=\sum\limits_{(i,j)\in\Z^2_0}a_{i,j}\partial_{i,j}U(z)+f(|U(z)|^2)U(z)\, ,
\end{equation}
where $z=n\cos \th+m\sin\th-\nu t$, $\nu\ne0$ and
$$
\partial_{i,j}U(z):=U(z+i\cos\th+j\sin\th)+U(z-i\cos\th-j\sin\th)-2U(z).
$$
We see that \eqref{2e2} has a very similar form like \eqref{e2}. So we can directly repeat the above arguments, where now instead of $\Phi(x)$ we get
$$
\Phi_\th(x):=\frac{4}{x}\sum\limits_{(i,j)\in\Z^2_0}a_{i,j}\sin^2\frac{x(i\cos\th+j\sin\th)}{2}.
$$
Set $\bar R_\th:=\sup_\R\Phi_\th$. Summarizing, Theorems \ref{th1} and \ref{th2} have the following analogies:

\begin{theorem}\label{2th1} Let (H1) hold and $T>0$, $\th\in[0,2\pi)$. Then for almost each $\nu \in \R\setminus\{0\}$ and any rational $r\in \Q\cap (0,1)$, there is a nonzero periodic traveling wave solution \eqref{2e2a} of \eqref{2e1} with $U\in C^1(\R,\C)$ satisfying \eqref{prop}.
Moreover, for any $\nu \in \R\setminus\{0\}$, there is at most a finite number of $\bar r_{1,\th},\bar r_{2,\th},\cdots,\bar r_{m_\th,\th}\in (0,1)$ such that equation
$$
-\nu=\Phi_\th\left(\frac{2\pi}{T}(\bar r_{j,\th}+k)\right)
$$
has a solution $k\in\Z$. Then for any $r\in (0,1)\setminus \{\bar r_{1,\th},\bar r_{2,\th},\cdots,\bar r_{m_\th,\th}\}$ there is a nonzero quasi periodic traveling wave solution \eqref{2e2a} of \eqref{2e1} with the above properties. In particular, for any $|\nu|>\bar R_\th$ and $r\in(0,1)$, there is such a nonzero quasi periodic traveling wave solution.\end{theorem}

\begin{theorem}\label{2th2}
Suppose $f\in C^2(\R_+,\R)$ with $f(0)=0$. If there are $\bar r_{1,\th}\in(0,1)$, $T>0$, $\th\in[0,2\pi)$ and $\nu\in\RR\Phi_\th\setminus\{0\}$ such that all integer number solutions $k_{1},k_{2},\cdots,k_{m_{1,\th}}$ of equation
$$
-\nu=\Phi_\th\left(\frac{2\pi}{T}(\bar r_{1,\th}+k)\right)
$$
are either nonnegative or negative, and $m_{1,\th}>0$. Then for
any $\ep>0$ small there are $m_{1,\th}$ branches of nonzero quasi
periodic traveling wave solutions \eqref{2e2a} of \eqref{2e1} with
$U_{j,\ep}\in C^1(\R,\C)$, $j=1,2,\cdots,m_{1,\th}$, and nonzero
velocity $\nu_{\ep}$ satisfying $U_{j,\ep}(z+T)=\eu^{2\pi \bar
r_1\imath}U(z)_{j,\ep}$, $\forall z\in \R$ along with $\nu_\ep\to
\nu$ and $U_{j,\ep}\rightrightarrows 0$ uniformly on $\R$ as
$\ep\to0$.
\end{theorem}

\begin{example}\label{2dexam1} We consider the discrete 2D Kac-Baker interaction kernel $a_{i,j}=\eu^{-|i|-|j|}$ for $(i,j)\in\Z^2_0$. Then $\sum\limits_{(i,j)\in\Z^2_0}\eu^{-|i|-|j|}=\frac{4\eu}{(e-1)^2}$ and
$$
\Phi_\th(x)=\left[\frac{(\eu+1)^2}{(\eu-1)^2}-\frac{(\eu^2-1)^2}{(1+\eu^2-2\eu\cos(x\cos\th))(1+\eu^2-2\eu\cos(x\sin\th))}\right]\frac{4}{x}.
$$
A numerical evaluation shows that function $(x,\th)\to \Phi_\th(x)$ has a maximum $\bar R\doteq9.75047$ at $x_0\doteq1.08205$ and $\th_0\doteq0.785398$. To justify this theoretically, we take $a=x\cos \th$ and $b=x\sin \th$ to transform $\Phi_\th(x)$ into
$$
\Phi(a,b)=\left[\frac{(\eu+1)^2}{(\eu-1)^2}-\frac{(\eu^2-1)^2}{(1+\eu^2-2\eu\cos a)(1+\eu^2-2\eu\cos b)}\right]\frac{4}{\sqrt{a^2+b^2}}\, .
$$
Note $\Phi(a,b)=\Phi(\pm a,\pm b)=\Phi(b,a)$. A numerical evaluation shows that function $\Phi_\th(a,b)$ has a maximum $\bar R\doteq9.75047$ at $a_0=b_0\doteq0.765123$ which correspond to $x_0$ and $\th_0$. On the other hand, if $a^2+b^2\ge 4$ then $\Phi(a,b)\le 2\frac{(\eu+1)^2}{(\eu-1)^2}\doteq9.36539<9.75047$, so $\Phi(a,b)$ achieves its maximum in the disc $D_2:=\left\{a^2+b^2\le 4\right\}$. Next, solving the system $\frac{\partial}{\partial a}\Phi(a,b)=\frac{\partial}{\partial b}\Phi(a,b)=0$ we derive $\frac{\sin a_0}{a_0}=\frac{\sin b_0}{b_0}$ at the maximum point $(a_0,b_0)\in D_2$, $a_0>0$, $b_0>0$. But the function $\frac{\sin w}{w}$ is decreasing on $[0,2]$, so $a_0=b_0$, and thus $\th_0=\pi/4$. An elementary but awkward calculus shows for function
$$
\Phi_{\pi/4}(x)=\left[\frac{(\eu+1)^2}{(\eu-1)^2}-\frac{(\eu^2-1)^2}{\left(1+\eu^2-2\eu\cos\left(x\frac{\sqrt{2}}{2}\right)\right)^2}\right]\frac{4}{x}
$$
with the graph on $[-20,20]$:
\begin{center}
    \includegraphics[clip,width=6cm]{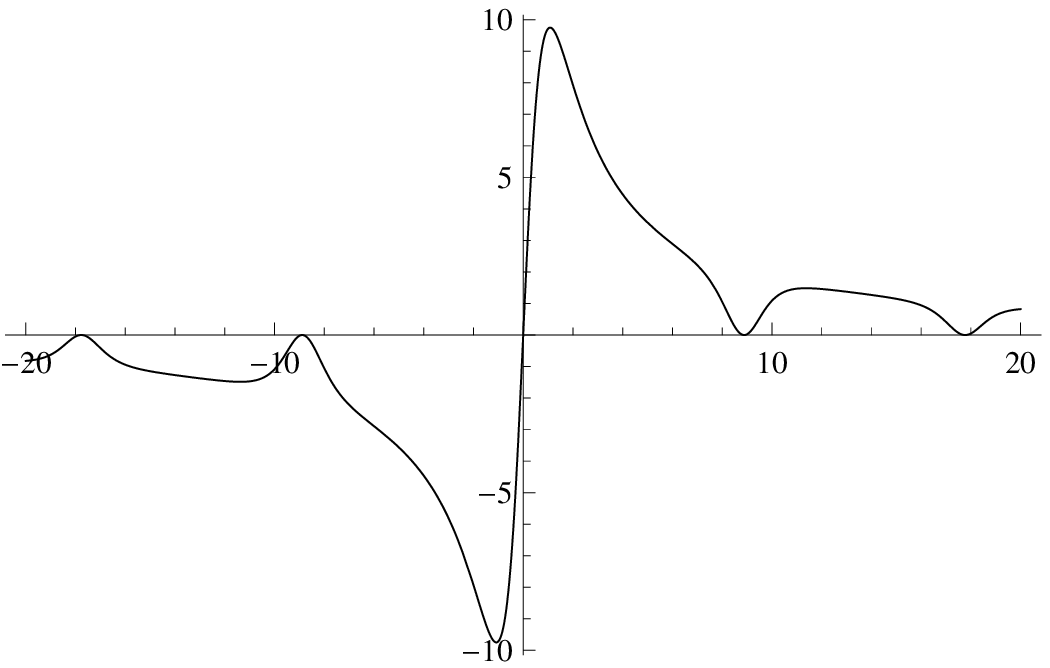}
\end{center}
that $x_0\in(0,2)$ is the only root of $\Phi_{\pi/4}'(x_0)=0$ on $(0,2)$, and then $\bar R=\Phi_{\pi/4}(x_0)$. So $\bar R$ is computed also analytically in this case.

Summarizing, Theorems \ref{2th1} and \ref{2th2} can be applied in
this case for any suitable nonzero $\nu$, and resonant traveling
waves with maximum velocities which are achieved in the diagonal
directions $\pm \th_0=\pm\pi/4$.\end{example}

Finally, it is now clear how to proceed to 3D DNLS or even to higher dimensional DNLS, so we omit further details.

\section{Traveling Waves with Frequencies}\label{more}

We could consider more general traveling wave solutions than above of forms
\begin{equation}\label{3e1}
\begin{gathered}
u_n(t)=U(n-\nu t)\eu^{\imath\om t},\\
u_{n,m}(t)=U(n\cos \th+m\sin\th-\nu t)\eu^{\imath\om t}\end{gathered}
\end{equation}
with velocity $\nu\ne 0$ and frequency $\om\ne0$ (see \cite{Pel1}). Then, there is a dispersion relation between the velocity $\nu$ and frequency $\om$ as follows. Inserting \eqref{3e1} into \eqref{e1} and \eqref{2e1}, respectively, we are interested in equations
\begin{equation}\label{3e2}
\begin{gathered}
-\nu \imath U'(z)=\sum\limits_{j\in\N}a_j\partial_jU(z)+\om U(z)+f(|U(z)|^2)U(z),\\
-\nu \imath U'(z)=\sum\limits_{(i,j)\in\Z^2_0}a_{i,j}\partial_{i,j}U(z)+\om U(z)+f(|U(z)|^2)U(z),\end{gathered}
\end{equation}
respectively. We see that \eqref{e2}, \eqref{2e2} and \eqref{3e2} are very similar, so we can repeat the above arguments to \eqref{3e2} when instead of $\Phi(x)$ and $\Phi_\th(x)$ now we have
\begin{equation}\label{3e3}
\Phi(x,\om):=\Phi(x)-\frac{\om}{x}\, ,\quad \Phi_\th(x,\om):=\Phi_\th(x)-\frac{\om}{x}\, ,
\end{equation}
respectively. Consequently, we have analogies of Theorems \ref{th1}, \ref{th2}, \ref{2th1} and \ref{2th2} to \eqref{3e2} but we do not state them since they are obvious.
\begin{example}\label{3exam1} We consider the discrete Kac-Baker interaction kernel from Example \ref{exam2}. Then
$$
\Phi(x,\om)=\frac{2\eu(\eu+1)(1-\cos x)}{(\eu-1)x(\eu^2+1-2\eu\cos x)}-\frac{\om}{x}\, .
$$
To be more concrete, we first take $\om=1$, and then $\Phi(x,1)$ has the graph on $[-4\pi,4\pi]$:
\begin{center}
    \includegraphics[clip,width=6cm]{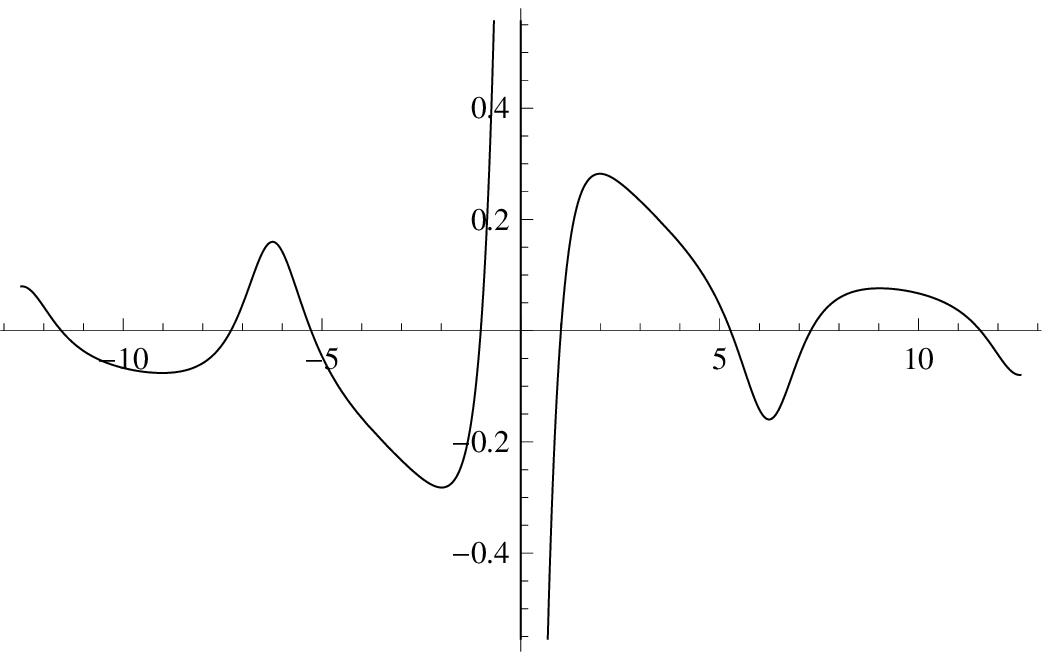}
\end{center}
with $\lim_{x\to 0_\pm}\Phi(x,1)=\mp\infty$. A numerical evaluation shows that function $\Phi(x,1)$ has a maximum $\bar R\doteq0.282071$ on $(0,\infty)$ at $x_0\doteq1.9905$. Consequently, the analogy of Theorem \ref{th1} can be applied now to any $\nu\ne0$ while the analogy of Theorem \ref{th2} can be applied for almost any $\nu\in\R\setminus[-0.282071,0.282071]$, while for nonzero $\nu\in[-0.282071,0.282071]$ could be problematic in general.

On the other hand for $\om=-1$, $\Phi(x,-1)$ has the graph on $[-4\pi,4\pi]$:
\begin{center}
    \includegraphics[clip,width=6cm]{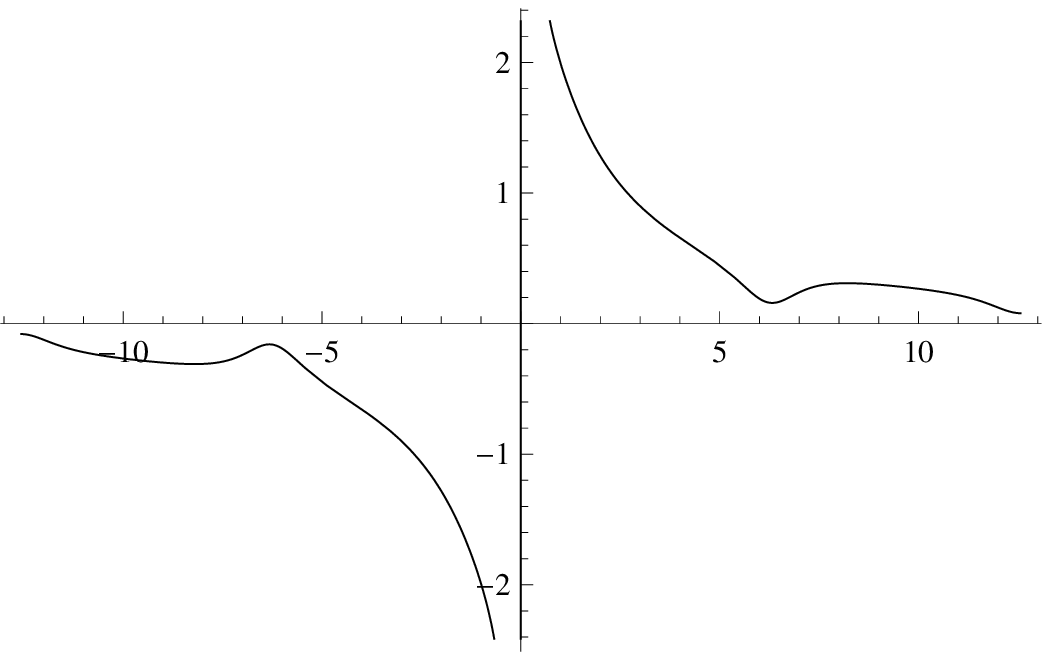}
\end{center}
with $\lim_{x\to 0_\pm}\Phi(x,-1)=\pm\infty$. Consequently, the analogy of Theorem \ref{th1} can again be applied now to any $\nu\ne0$ while the analogy of Theorem \ref{th2} can now be applied for almost any $\nu\ne0$.  Of course now we have totally different situations than in Example \ref{exam2} for traveling waves without frequencies by comparing the above graphs with that one in Example \ref{exam2}.\end{example}

\subsection{Acknowledgements} Michal Fe\v ckan is partially supported by the Grants VEGA-MS 1/0098/08 and VEGA-SAV 2/7140/27. Vassilis Rothos is partially supported by Research Grant-International Relations of AUTH.

\end{document}